\documentclass{article}

\usepackage{amssymb, amsmath, amsthm, tikz, verbatim, graphicx,lmodern,indentfirst,enumerate,color}

\usepackage[titletoc,page]{appendix}
\usepackage[margin=1.4in]{geometry}

\usepackage{natbib}

\newtheorem{thm}{Theorem}[section]

\newtheorem{prop}[thm]{Proposition}

\newtheorem{oprob}[thm]{Open Problem}
\newtheorem{prob}[thm]{Problem}

\theoremstyle{definition}

\theoremstyle{definition}
\newtheorem{defn}[thm]{Definition}

\theoremstyle{remark}
\newtheorem{rem}[thm]{Remark}

\numberwithin{equation}{section}

\makeatletter
\newcommand{\rmnum}[1]{\romannumeral #1}
\newcommand{\Rmnum}[1]{\expandafter\@slowromancap\romannumeral #1@}
\makeatother
\begin{document}

\title{A geometric approach to the transfer problem for a finite number of traders}
\author{Tomohiro Uchiyama\\
National Center for Theoretical Sciences, Mathematics Division\\
No.~1, Sec.~4, Roosevelt Rd., National Taiwan University, Taipei, Taiwan\\
\texttt{email:t.uchiyama2170@gmail.com}}
\date{}
\maketitle 

\begin{abstract}
We present a complete characterization of the classical transfer problem for an exchange economy with an arbitrary finite number of traders. Our method is geometric, using an equilibrium manifold developed by Debreu, Mas-Colell, and Balasko. We show that for a regular equilibrium the transfer problem arises if and only if the index at the equilibrium is $-1$. This implies that the transfer problem does not happen if the equilibrium is Walras tatonnement stable. Our result generalizes Balasko's analogous result for an exchange economy with two traders.  
\end{abstract}

\section{Introduction}
In this paper, we study the following classical \emph{transfer problem}~\citep{Samuelson-transfer-EconomicJournal, Samuelson-transfer-EconomicJournal2}:
\begin{prob}
Suppose that in the world there are $n$ countries trading $l$ goods where $n$ and $l$ are arbitrary natural numbers. Then does it happen that after a country, say country A, gives away some of her endowment to other countries, country A's utility level goes up? If it happens, characterize it. 
\end{prob}   
Samuelson considered an exchange economy with $n=2$ and $l=2$ (with or without various trade impediments). He also hinted that (without proof) the transfer problem is closely related to the economic stability. Balasko verified that Samuelson's intuition was correct. He showed that in a smooth exchange economy (without any trade impediments) the tranfer problem arises at a regular equilibrium if and only if the index at the equilibrium is $-1$, in other words, this pathology happens only when the equilibrium is locally Walras tatonnement unstable. Using the theory of equilibrium manifold, Balasko initially showed this for $n=2$ and $l=2$~\citep{Balasko-transfer-IER}, then for $n=2$ and an arbitrary finite number $l$~\citep{Balasko-transfer-TheoreticalEcon}.  

The main purpose of this paper is to extend Balasko's result for arbitrary finite numbers $n$ and $l$. Roughly speaking, we show that
\begin{thm}\label{mainrough}
For a smooth exchange economy with arbitrary finite numbers of traders (countries) and goods (without any trade impediments), the transfer problem arises at a regular equilibrium if and only if the index at the equilibrium is $-1$. In particular, the transfer problem arises only when the regular equilibrium is locally Walras tatonnement unstable. 
\end{thm}

For the precise formulation of the transfer problem, see Definition~\ref{transferproblem}. Our model suits better than Balasko's model to the current world economy where many countries (insted of just $2$ countries) trade large amount of goods. We do not consider any trade impediments. This simplifies our model, and also this is a reasonable assumption considering that various trade impediments have become negligible in the modern global world economy. Some might disagree with us: we leave the issue about trade costs for a future work.  

Our approach is (differential) geometric: we use the theory of equilibrium manifold in~\citep{Debreu-Finite-Econometrica},~\citep{Mas-Colell-differential-book},~\citep{Balasko-manifolds-book, Balasko-value-book}.
We restrict our equilibrium analysis to the set of regular equilibria. This is not a strong assumption: the set of regular equilibria is an open subset with the full measure in an equilibrium manifold~\citep[Prop.~8.10]{Balasko-value-book}.

We kept the mathematical requirement minimum. Readers just need to know basic differential geometry (or differential topology), say, Milnor's little book~\citep{Milnor-topology-book} or Guillemin-Pollack's classic~\citep{Guillemin-topology-book} is more than enough. (Actually just knowing the definition of a differentiable manifold and the statement of the transversality theorem suffices.)   
All necessary backgrounds on equilibrium manifolds, economic stability, and the index theorem are contained in the paper. Our references for general equilibrium theory and equilibrium manifolds are~\citep{Balasko-manifolds-book, Balasko-value-book} and~\citep{Mas-Colell-differential-book}.

Here is the structure of the paper. In Section 2, we set out notation and explain our economic model. Then in Section 3, we review the theory of equilibrium manifold and formulate the transfer problem in a precise way. In Section 4, we quickly review the relationship between economic stability and the index theorem. In Section 5, we thoroughly analyse our model and prove a key result Proposition~\ref{keyprop}. Then in Section 6, we apply all results from previous sections to prove the main result Theorem~\ref{mainthm}. Finally in Section 7,  we state several open problems.

\section{An exchange economy}
We denote by $\mathbb{N}$, $\mathbb{R}$, and $\mathbb{R}_{>0}$ the set of natural numbers, the set of real numbers, and the set of positive real numbers respectively. Fix $l\in \mathbb{N}$. We set the commodity space $X:=\mathbb{R}^l_{>0}$ (the strictly positive orthant of $\mathbb{R}^l$). We set the $l$-th goods as a a numeraire: if $p=(p_1,\cdots ,p_l)\in X$ is a price vector, then $p_l=1$. We assume that all prices are strictly positive. The set of price vectors is $S:=\mathbb{R}^{l-1}_{>0}\times \{1\}$. We have $n\in \mathbb N$ traders (countries). Each trader $i\in \{1,\cdots, n\}$ is endowed with a goods vector $\omega_i\in X$. We assume that our economy is an exchange economy with no production: the total resources $r=\sum_{i=1}^{n}\omega_i$ of the economy is fixed. We write $\Omega=\{\omega=(\omega_1,\cdots, \omega_n)\in X^n \mid \sum_{i=1}^{n}\omega_i=r \}$ for the space of endowment of the economy. 

We assume that each trader's preferences is represented by a smooth utility function $u_i: X\rightarrow \mathbb{R}$. For $z_i\in X$, we write $I_{z_i}:=\{ y_i\in X\mid u_i(y_i)=u_i(z_i)\}$ (the indifference surface for trader $i$ going through the point $z_i$). We assume that each $u_i$ satisfies the following standard hypotheses for the differential equilibrium analysis~\citep{Mas-Colell-differential-book},~\citep{Balasko-manifolds-book}. For any $x_i \in X$,
(\rmnum{1})~smooth monotonicity: $D u_i(x_i)\in \mathbb{R}^l_{>0}$;
(\rmnum{2})~smooth strict quasi-concavity: the Hessian $D^2 u_i(x_i)$ is negative definite on the tangent plane to $I_{x_i}$;
(\rmnum{3})~$I_{x_i}$ is closed in $X$. Moreover, we extend $u_i$ to $x_i=0$ by setting $u_i(0)=\textup{inf}_{x_i\in X}u_i(x_i)$. 

Given a price vector $p\in S$ and an budget $w_i=p\cdot \omega_i\in \mathbb{R}_{>0}$, each trader maximizes the utility $u_i(x_i)$ under the budget constraint $p\cdot x_i \leq w_i$. Then there exists a unique solution to this utility maximization problem, and we obtain a smooth demand function $f_i:S\times \mathbb{R}_{>0}\rightarrow X$ satisfying Walras's law: $p\cdot    
f_i(p,w_i)=w_i$. Note: Do not confuse $w_i$ with $\omega_i$. We are following Debreu's (and Balasko's) notation~\citep{Debreu-Finite-Econometrica},~\citep{Balasko-transfer-TheoreticalEcon}.

\section{The equilibrium manifold}
For $p\in S$ and $\omega\in \Omega$, let $z(p,\omega):=\sum_{i=1}^{n}f_i(p,p\cdot \omega_i)-r$. Then $z(p,\omega)$ is the (social) excess demand associated with $(p,\omega)$. It is clear that the function $z: S\times X\rightarrow \mathbb{R}^l$ is smooth. 
\begin{defn}
We call $(p,\omega)\in S\times X$ is an \emph{equilibrium} if $z(p,\omega)=0$. The subset $E$ of $S\times \Omega$ defined by $z(p,\omega)=0$ is called the \emph{equilibrium manifold}. 
\end{defn}
Note that the equilibrium manifold $E$ is a real smooth manifold in the usual sense in differential geometry~\citep{Milnor-topology-book},~\citep{Guillemin-topology-book}. See~\citep[Prop.~2.4.1]{Balasko-manifolds-book} for a proof.

For $(p,\omega)\in S\times X$, we denote by $z_{-l}(p,\omega)\in \mathbb{R}^{l-1}$ the vector defined by the first $l-1$ components of $z(p,\omega)$. Since each trader satisfies Walras' law at an equilibrium, we have $z(p,\omega)=0$ if and only if $z_{-l}(p,\omega)=0$. (Thus it is enough to consider $z_{-l}(p,\omega)$.) Let $p\in S$. We write $p_{-l}$ for the vector defined by the first $l$-coordinates of $p$. (The $l$-th coordinate is fixed by the numeraire assumption anyway.) Recall that an equilibrium $(p,\omega)\in E$ is \emph{regular} if the Jacobian matrix $J(p,\omega):=D z_{-l}(p,\omega)/{D p_{-l}}$ is invertible~\citep{Mas-Colell-differential-book}~\citep{Balasko-manifolds-book}. 

Suppose that $(p,\omega)\in E$ is a regular equilibrium. Then by applying the implicit function theorem to the equation $z_{-l}(p,\omega)=0$, we can express $p$ in terms of $\omega$. Then by~\cite[Prop.~7.2]{Balasko-value-book}, there exist a neighborhood $U$ of $\omega$, a neighborhood $V$ of $(p,\omega)$, and a smooth map $s:U\rightarrow S$ such that the map $\sigma: U\rightarrow V$ defined by $\sigma(x)=(s(x),x)$ is a diffeomorphism between $U$ and $V$. Following~\citep{Balasko-transfer-TheoreticalEcon}, we call $\sigma$ (or $s$) the \emph{local equilibrium selection map (the local equilibrium price selection map) associated with $(p,\omega)$}. Now we are ready to state the transfer problem in a precise way:
\begin{defn}\label{transferproblem}
We say that there is a \emph{transfer problem} at a regular equilibrium $(p,\omega)\in E$ if there exists an endowment vector $\omega'\in U$ and a trader $i\in \{1,\cdots, n\}$ such that 
\begin{enumerate}
\item{$\omega_i\leq \omega'_i$ and $\omega_i\neq \omega'_i$,}
\item{$u_i(f_i(s(\omega'),s(\omega')\cdot \omega'_i))>u_i(f_i(s(\omega),s(\omega)\cdot \omega_i))$}
\end{enumerate}
where $s:U\rightarrow S$ is the local equilibrium price selection map associated with $(p,\omega)$ defined above. 
\end{defn} 
Note that the definition of the transfer problem at $(p,\omega)\in E$ requires $(p,\omega)$ to be regular (since otherwise we do not have the function $s$). However this is not a strong restriction: the set of regular equilibria in $E$ is open with the full measure of $E$; see~\cite[Prop.~8.10]{Balasko-value-book}.

\section{Economic stability and the index theorem}
In this short section, we review the definition of the index at an equilibrium $(p,\omega)$ and its relationship with economic stability to make this paper self-contained. Our references here are~\citep{Arrow-dynamics-book},~\citep{Smale-dynamics-book}, and~\citep{Guillemin-topology-book}. Recall that the \emph{index} of a regular equilibrium $(p,\omega)\in E$ is $1$ (or $-1$) if the sign of $(-1)^{l-1}\textup{det}(J(p,\omega))$ is positive (or negative respectively). It is well-known that if $(p,\omega)\in E$ is a Walras tatonnement locally stable equilibrium, then all eigenvalues of $J(p,\omega)$ must have strictly negative real parts. Thus the index of such $(p,\omega)\in E$ must be $1$. Note that the converse fails: there exists an equilibrium $(p,\omega)\in E$ such that the index of $(p,\omega)$ is $1$ but $(p,\omega)$ is \textbf{not} Walras tatonnement locally stable. See the references above for more on this.

\section{A geometric approach to equilibrium analysis}
In this section, we follow the argument in~\citep[Sec.~3]{Balasko-transfer-TheoreticalEcon} closely.
Define the \emph{price-income space} $H(r)$ associated with the fixed social endowment $r$ by $H(r):=\{(p,w)\in S\times \mathbb{R}^n_{>0}\mid \sum_{i=1}^{n}w_i=p\cdot r\}$ where $w$ is the income vector $(w_1,\cdots,w_n)$. From the definition, it is clear that $H(r)$ is an $(l+n-2)$-dimensional hyperplane in $S\times \mathbb{R}^n_{>0}$. We set the coordinates for $H(r)$ by $(p_1,\cdots,p_{l-1},w_1,\cdots,w_{n-1})=(p_{-l},w_{-n})\in \mathbb{R}^{l+n-2}_{>0}$. Note that we omitted the last coordinate of $w$ since $w_n$ is determined by $w_n=p\cdot r-\sum_{i=1}^{n-1}w_{i}$. 

Let $B(r):=\{(p,w)\in H(r)\mid \sum_{i=1}^{n}f_i(p,w_i)=r\}$. We call $B(r)$ the \emph{section manifold}. This is a smooth $(n-1)$-dimensional submanifold of $H(r)$ by~\cite[Prop.~5.4.1]{Balasko-manifolds-book}.
In the following, we give a concrete parametrization of $B(r)$ that comes handy for our geometric equilibrium analysis.

Let $P(r)$ be the set of Pareto optimal allocations. First, we parametrize $P(r)$. Let $U(r)$ be the set of vectors $(u_1,\cdots,u_{n-1})=u_{-n}\in\mathbb{R}^{n-1}$ of the first $n-1$ traders' feasible utility levels, that is 
\begin{alignat*}{2}
U(r):=&\{u_{-n}\in \mathbb{R}^{n-1}\mid \textup{there exists }(x_1,\cdots,x_{n-1})\in X^{n-1} \\
&\textup{ such that }u_1=u_1(x_1),\cdots,u_{n-1}=u_{n-1}(x_{n-1}) \textup{ and }\sum_{i=1}^{n-1}x_i\leq r\}.
\end{alignat*}
Given $u_{-n}\in U(r)$ we obtain a Pareto optimal allocation $x=(x_1,\cdots,x_n)\in X^n$ by solving the following constrained optimization problem:
\begin{equation*}
\max_{x\in X^n}u_n(x_n) \textup{ subject to } u_{-n}=u_{-n}(x_{-n})\textup{ for some } x_{-n}\in U(r) \textup{ and } \sum_{i=1}^{n}x_i=r. 
\end{equation*}
We write $x(u_{-n})=(x_1(u_{-n}),\cdots,x_n(u_{-n}))$ for the solution of this optimization problem. Under our assumptions on smooth preferences, there exists a unique price vector $p(u_{-n})\in S$ that supports $x(u_{-n})$. Note that $p(u_{-n})$ is parallel to the gradiant vectors $D u_i(x_i(u_{-n}))$. We can express all Pareto optimal allocations as $x(u_{-n})$ by varying $u_{-n}$ in $U(r)$.   

Here is a crucial observation: the point $M(u_{-n}):=(p(u_{-n}),p(u_{-n})\cdot x_1(u_{-n}),\cdots, p(u_{-n})\cdot x_n(u_{-n}))\in H(r)$ belongs to $B(r)$ for any $u_{-n}\in U(r)$. Conversely, any point $(p,w)\in B(r)$ is associated with a unique point $u_{-n}\in U(r)$ via $u_{-n}=\left(u_1(f_1(p,w_1),\cdots,u_{n-1}(f_{n-1}(p,w_{n-1}))\right)$. Thus, $u_{-n}\in U(r)$ parametrizes the section manifold $B(r)$ as well as the set of Pareto optimal allocations $P(r)$. 

For the map $M:U(r)\rightarrow B(r)$, we write $\frac{\partial M}{\partial u_i}$ for the partial derivative of $M$ with respect to $u_i$ for $i\in\{1,\cdots,n-1\}$. 
We set a positive orientation for $B(r)$ via $(\frac{\partial M}{\partial u_1},\cdots,\frac{\partial M}{\partial u_{n-1}})$. For $u_{-n}\in U(r)$, the vector $\frac{\partial M}{\partial u_{i}}(u_{-n})$ represents the direction of increasing the utility level for trader $i$ in a neighborhood of $M(u_{-n})$.

For $\omega\in \Omega$, we define the \emph{budget space} $A(\omega)$ associated with the endowment $\omega$ by
$A(\omega):=\{(p,w)\in H(r)\mid w_i=p\cdot \omega_i \textup{ for }i\in \{1,\cdots,n-1\}\}$.
Then $A(\omega)$ is a linear space being the intersection of hyperplanes defined by $w_i=p\cdot \omega_i$ in the hyperplane $H(r)$.

It is easy to see that $(p,\omega)\in S\times \Omega$ is an equilibrium if and only if $(p,p\cdot\omega_1,\cdots,p\cdot\omega_n)\in A(\omega)\cap B(r)$. The point is that we split the equilibrium analysis into the analysis of two separate parts $A(\omega)$ and $B(r)$. The first part $A(\omega)$ is a linear space. Thus all nonlinearities of the equilibrium equation $z(p,\omega)=0$ is captured by $B(r)$. Note that $B(r)$ does not depend on $\omega$: all $\omega$-terms are in $A(\omega)$.

Using the coordinate system $(p_{-l},w_{-n})$ for $H(r)$ and the equations $w_i=p\cdot \omega_i$, we see that $A(\omega)$ is parallel to the $(l-1)$-dimensional vector space whose basis $\{a_1(\omega),\cdots, a_{l-1}(\omega)\}$ (in terms of this coordinate system) is given by
\begin{alignat*}{2}
a_1(\omega)&:=(1,0,\cdots,0,\omega_{1}^{1},\omega_{2}^{1},\cdots, \omega_{n-1}^{1}),\\
a_2(\omega)&:=(0,1,0,\cdots,0,\omega_{1}^{2},\omega_{2}^{2},\cdots, \omega_{n-1}^{2}),\\
     &\cdots \\  
a_{l-1}(\omega)&:=(0,\cdots,0,1,\omega_{1}^{l-1},\omega_{2}^{l-1},\cdots, \omega_{n-1}^{l-1}),\\
&\textup{ where }\omega_i^{j} \textup{ is trader $i$'s endowment of goods $j$}.
\end{alignat*}  
We set the positive orientation of $A(\omega)$ via $(a_1(\omega),\cdots, a_{l-1}(\omega))$. Clearly the set of vectors $a_i(\omega)$ is linearly independent, so the dimension of $A(\omega)$ is $l-1$. 

Pick a regular equilibrium $(p,\omega)\in E$. Then $(p,p\cdot\omega_1, \cdots, p\cdot\omega_n)$ is the corresponding point in the price-income space $H(r)$. Let $u_i(p,\omega):=u_i(f_i(p,p\cdot \omega_i))$ for $i\in\{1,\cdots,n-1\}$. Since $(p,\omega)$ is regular, smooth submanifolds $A(\omega)$ and $B(r)$ of $H(r)$ intersect transversally at $(p,p\cdot\omega_1, \cdots, p\cdot\omega_n)$. By the celebrated transversality theorem~\citep[Chap.1.I]{Mas-Colell-differential-book}, this is equivalent to
\begin{equation*}
\Delta(p,\omega):=\textup{det}\left(a_1(\omega),\cdots,a_{l-1}(\omega), \frac{\partial M}{\partial u_{1}}(u_{-n}(p,\omega)),\cdots,\frac{\partial M}{\partial u_{n-1}}(u_{-n}(p,\omega))\right)\neq 0.
\end{equation*}  

Here is the main result of this section.
\begin{prop}\label{keyprop}
Let $(p,\omega)\in E$ be regular. Then the index of $(p,\omega)$ is $1$ (or $-1$) if $\Delta(p,\omega)$ is positive (or negative respectively).
\end{prop}
\begin{proof}
Although we could prove the proposition by a (not illuminating) tedious direct computation, we avoid it by utilizing the parametrization of $B(r)$ via $U(r)$ above.
The following argument shows that $\Delta(p,\omega)$ has the opposite sign of $J(p,\omega)$ for any regular $(p,\omega)\in E$, which is sufficient for our purpose. 

For $u_{-n}\in U(r)$. let $M(u_{-n})=(p(u_{-n}),w_1(u_{-n}),\cdots,w_n(u_{-n}))\in B(r)$. Now let $\omega_i(u_{-n}):=f_i(p(u_{-n}),w_i(u_{-n}))$ for $i\in\{1,\cdots, n\}$, and write $\omega(u_{-n})$ for $(\omega_1(u_{-n}),\cdots,\omega_n(u_{-n}))$. Then $(p(u_{-n}),\omega(u_{-n}))$ is a no-trade equilibrium. By~\cite[Prop.~8.2]{Balasko-value-book} every no-trade equilibrium is regular, so $(p(u_{-n}),\omega(u_{-n}))$ is regular.  

Note that $A(\omega(u_{-n}))$ varies continuously as $u_{-n}$ moves around in $U(r)$. Thus the function $\delta: U(r)\rightarrow \mathbb{R}$ defined by $\delta(u_{-n}):=\Delta(p(u_{-n}),\omega(u_{-n}))$ is continuous. We see that $\delta$ never take the value $0$ for any $u_{-n}\in U(r)$ since every no-trade equilibrium is regular. So it is enough to check the sign of $\delta$ for any particular value of $u_{-n}$ since $U(r)$ is connected. (We are using the same argument as in the proof of ~\citep[Lem.~1]{Balasko-transfer-TheoreticalEcon}.)

Let $\underline\omega=(0,0,\cdots, 0, r)\in X^n$. Then $u_{-n}(\underline\omega):=(u_1(0), u_2(0), \cdots, u_{n-1}(0))$. Thus
\begin{alignat*}{2}
\Delta\left(p(u_{-n}(\underline\omega)),\underline\omega\right):=&\textup{det}\left(
e_1, e_2, \cdots, e_{l-1}, 
\frac{\partial M}{\partial u_{1}}(u_{-n}(\underline\omega)),\cdots,\frac{\partial M}{\partial u_{n-1}}(u_{-n}(\underline\omega))\right)
\end{alignat*}
where $e_i$ is  the column vector whose $i$-th coordinate is $1$ and all other coordinates are  zero.

Let ${\Delta}'\left(p(u_{-n}(\underline\omega)),\underline\omega\right)$ be the submatix of $\Delta\left(p(u_{-n}(\underline\omega)),\underline\omega\right)$ formed by the last $(n-1)$ rows and columns of $\Delta\left(p(u_{-n}(\underline\omega)),\underline\omega\right)$. Then

\begin{equation*}
\delta(u_{-n}(\underline\omega))>0 \textup{ if and only if } \textup{det}\left({\Delta}'(p(u_{-n}(\underline\omega)),\underline\omega)\right)>0.
\end{equation*}
Note that the vector $\frac{\partial M}{\partial u_{i}}(u_{-n}(\underline\omega))$ for $i\in \{1,\cdots, n-1\}$ is the direction of increasing trader $i$'s utility level holding other trader's utility level fixed at $u_j(0)$ for all $j\in \{1,\cdots, n-1\}\backslash \{i\}$. Since the income of all traders $i\in \{1,\cdots, n-1\}$ is $0$ at $(p(u_{-n}(\underline\omega)), \underline\omega)$, we must have 
\begin{equation*}
\frac{\partial M}{\partial u_{i}}(u_{-n}(\underline\omega)) = (0,0,\cdots, 0, x, 0,\cdots, 0) \textup{ where $x$ is the $i$-th component and $x\in \mathbb{R}_{>0}$}.
\end{equation*} 
Therefore we clearly have $\textup{det}\left({\Delta}'(p(u_{-n}(\underline\omega)),\underline\omega)\right)>0$, and we are done. 
\end{proof}

\begin{rem}
We have shown that the index at $(p,\omega)\in E$ is same as the \emph{intersection number} of $A(\omega)$ and $B(r)$ at the corresponding point $(p, p\cdot\omega_1,\cdots, p\cdot \omega_{n})\in H(r)$; see~\citep[p.96]{Guillemin-topology-book} for the definition of the intersection number of two smooth manifolds. 
\end{rem}

\section{The transfer problem}
Let $i\in\{1,\cdots, n-1\}$. Let $\omega, \omega'\in \Omega$. Following~\cite[Sec.~4]{Balasko-transfer-TheoreticalEcon}, we say that \emph{the hyperplane $A(\omega)$ lies below the hyperplane $A(\omega')$ for trader $i$} if $p\cdot \omega_i < p\cdot \omega_i'$ is satisfied for all $p\in S$. It is clear that $A(\omega)$ lies below $A(\omega')$ if and only if $\omega_i\lneqq\omega_i'$; see~\cite[Lem.~2]{Balasko-transfer-TheoreticalEcon} for a proof.

Now, we interpret the local equilibrium price selection map in Definition~\ref{transferproblem} using $A(\omega)$ and $B(r)$. If $(p,\omega)\in E$ is regular, then $A(\omega)$ and $B(r)$ intersect transversally at $(p,p\cdot \omega_1, \cdot, p\cdot \omega_{n})$ in the price-income space $H(r)$. Then if $\omega'\in \Omega$ is sufficiently closed to $\omega$, $A(\omega')\cap B(r)$ contains a unique point in some small neighborhood of $(p,p\cdot \omega_1,\cdots, p\cdot \omega_n)\in H(r)$. 
Thus, in some neighborhood of $\omega$ we obtain a map $s:\omega'\rightarrow (p(\omega'),p\cdot \omega'_1,\cdots, p\cdot \omega'_{n})$ such that $s(\omega)=(p,p\cdot\omega_1,\cdots,p\cdot\omega_n)$. Now composing this map with the projection $ 
(p(\omega'),p\cdot \omega'_1,\cdots, p\cdot \omega'_{n})\rightarrow p(\omega')$ gives the local equilibrium price selection map in Definition~\ref{transferproblem}.

\begin{thm}\label{mainthm}
A regular equilibrium $(p,\omega)$ features the transfer problem if and only if the index of $(p,\omega)$ is $-1$. In particular, if $(p,\omega)$ is Walrus tatonnement locally stable, the transfer problem does not arise. 
\end{thm}  
\begin{proof}
Let $(p,\omega)\in E$ with $b(p,\omega):=(p,p\cdot \omega_1,\cdots,p\cdot\omega_{n})\in A(\omega)\cap B(r)$. 
Let $u_i(p,\omega)=u_i(p, p\cdot \omega_i)$ for $i\in \{1,\cdots, n-1\}$. Then $M(u_1(p,\omega),\cdots, u_{n-1}(p,\omega))=(p,p\cdot \omega_1,\cdots,p\cdot\omega_{n})\in H(r)$ ($M$ is defined in Section 5). Let $\underline\omega=(0,0,\cdots,0,r)\in \Omega$ and $\overline\omega=(r,0,0,\cdots,0)\in \Omega$.
Set
\begin{alignat*}{2} 
t(\underline \omega):=&M(u_1(0),\cdots, u_{n-1}(0)), \\
t(\overline \omega):=&M(u_1(r),\cdots, u_{n-1}(0)).
\end{alignat*}
Since $B(r)$ is a smooth connected manifold and parametrized by $u_{-n}\in U(r)$, there exist curves $C_1$ (and $C_2$) from $t(\underline\omega)$ to $b(p,\omega)$ (and from $b(p,\omega)$ to $t(\overline\omega)$ respectively) where points on $C_1$ (those on $C_2$) correspond to the utility level of trader $1$ lower than or equal to $u_1(p,\omega)$ (higher than or equal to $u_1(p,\omega)$ respectively). 

If the intersection number of $A(\omega)$ and $B(r)$ is $-1$ at $b(p,\omega)$, there exists a neighborhood $U\subset H(r)$ of $b(p,\omega)$ such that all points in $U\cap C_1$ is above $A(\omega)$ for trader $1$. Likewise, all points in $U\cap C_2$ is below $A(\omega)$. 

Let $U'$ be an open neighborhood of $\omega$ where the local equilibrium price selection map $s: U'\rightarrow P$ is defined. Without loss, we choose $U'$ such that the image of the projection of $U$ onto $\Omega$ is contained in $U$. Then if $\omega'\in U'$ and $\omega_1'\lneqq \omega_1$, $A(\omega')$ lies below $A(\omega)$ for trader $1$. So $b(s(\omega'), \omega')$ lies below $A(\omega)$ for trader $1$. Then $b(s(\omega'), \omega')$ belongs to the curve $C_2$. Thus we have $u_1(f_1(s(\omega'),s(\omega')\cdot \omega'_1))> u_1(f_1(s(\omega),s(\omega)\cdot \omega_1))$. 

Similarly we can show that if the intersection number of $A(\omega)$ and $B(r)$ is $1$ at $b(p,\omega)$, then  $u_1(f_1(s(\omega'),s(\omega')\cdot \omega'_1))< u_1(f_1(s(\omega),s(\omega)\cdot \omega_1))$ for $\omega'\in U'$ with $\omega'_1<\omega_1$. 

Note: although the proof shows there exists the transfer problem which improve trader $1$'s utility level after he/she gives away some of his/her endowment, there is nothing special for trader $1$. A similar argument works for any trader $i$.    
 \end{proof}

\begin{rem}
Suppose that we are at a no-trade equilibrium $(p,\omega)\in E$. In~\citep{Balasko-value-book}, Balasko showed that the set of regular equilibria is partitioned into a finite number of path-connected components and the index is constant in each component. Moreover he showed that the set of no-trade equilibria is contained in the unique component of index $1$. Therefore the transfer problem arises only for a large amount of trades. 
\end{rem}

\section{Open questions}
\begin{oprob}
Our model does not have a production sector. It would be interesting to see how the result changes if we introduce a production sector into the model.
\end{oprob}
\begin{oprob}
We have avoid the issue of trade costs. It is natural to ask: what happens if we introduce various kinds of trade impediments? Some discussions are in~\citep{Samuelson-transfer-EconomicJournal2}.  
\end{oprob}
\section*{Acknowledgements}
This research was supported by a postdoctoral fellowship at the National Center for Theoretical Sciences at the National Taiwan University.  
\bibliography{econbib}

\end{document}